\documentclass[11pt]{article}

\usepackage[letterpaper,margin=1in]{geometry}
\usepackage{amsmath, amssymb, amsthm, amsfonts}
\usepackage{bbm}
\usepackage{bm}

\usepackage{subcaption}

\usepackage{ifthen}
\usepackage{comment}
\usepackage{tikz}
\usetikzlibrary{positioning,decorations.pathreplacing}
\usepackage{comment}

\usepackage{cite}
\usepackage{appendix}
\usepackage{graphicx}
\usepackage{color}
\usepackage{algorithm}
\usepackage[noend]{algorithmic}
\usepackage{epstopdf}
\usepackage{wrapfig}
\usepackage{paralist}
\usepackage[textsize=tiny]{todonotes}

\usepackage{framed}
\usepackage[framemethod=tikz]{mdframed}
\usepackage[bottom]{footmisc}
\usepackage{enumitem}
\setitemize{noitemsep,topsep=3pt,parsep=3pt,partopsep=3pt}
\usepackage[font=small]{caption}
\usepackage{xspace}

\newtheorem{theorem}{Theorem}[section]
\newtheorem{lemma}[theorem]{Lemma}
\newtheorem{meta-theorem}[theorem]{Meta-Theorem}
\newtheorem{claim}[theorem]{Claim}
\newtheorem{remark}[theorem]{Remark}
\newtheorem{corollary}[theorem]{Corollary}

\newtheorem{definition}[theorem]{Definition}

\definecolor{darkgreen}{rgb}{0,0.5,0}
\usepackage{hyperref}
\hypersetup{
    unicode=false,          
    colorlinks=true,        
    linkcolor=red,          
    citecolor=darkgreen,        
    filecolor=magenta,      
    urlcolor=cyan           
}
\usepackage[capitalize]{cleveref}

\newcommand{\eps}{\varepsilon}

\newcommand{\poly}{\operatorname{\text{{\rm poly}}}}

\newcommand{\exclude}[1]{}

\newcommand{\FullOrShort}{full}

\ifthenelse{\equal{\FullOrShort}{full}}{

  \newcommand{\fullOnly}[1]{#1}	
	\newcommand{\tempfullOnly}[1]{#1}
  \newcommand{\shortOnly}[1]{}
  
  }{

    \newcommand{\fullOnly}[1]{}
		\newcommand{\tempfullOnly}[1]{}
		\newcommand{\shortOnly}[1]{#1}
    \newcommand{\IncludePictures}[1]{}
   
  }

\begin{document}

\date{}

\title{Synchronization Strings: List Decoding for Insertions and Deletions\footnote{Supported in part by a Simons Investigator Award, NSF grants CCF-1565641, CCF-1715187, CCF-1527110, CCF-1618280, and NSF CAREER award CCF-1750808.
}}	

\author{Bernhard Haeupler\\Carnegie Mellon University\\ \texttt{haeupler@cs.cmu.edu} \and
Amirbehshad Shahrasbi\\Carnegie Mellon University\\ \texttt{shahrasbi@cs.cmu.edu} \and
Madhu Sudan\\Harvard University\\ \texttt{madhu@cs.harvard.edu}}

\allowdisplaybreaks

\maketitle

\begin{abstract}
We study codes that are list-decodable under insertions and deletions (``insdel codes''). Specifically, we consider the setting where, given a codeword $x$ of length $n$ over some finite alphabet $\Sigma$ of size $q$, $\delta\cdot n$ codeword symbols may be adversarially deleted and $\gamma\cdot n$ symbols may be adversarially inserted to yield a corrupted word $w$. A code is said to be list-decodable if there is an (efficient) algorithm that, given $w$, reports a small list of codewords that include the original codeword $x$. Given $\delta$ and $\gamma$ we study what is the rate $R$ for which there exists a constant $q$ and list size $L$ such that there exist codes of rate $R$ correcting $\delta$-fraction insertions and $\gamma$-fraction deletions while reporting lists of size at most $L$.

Using the concept of {\em synchronization strings}, introduced by the first two authors [Proc. STOC 2017], we show some surprising results. We show that for every $0\leq \delta < 1$, every $0 \leq \gamma < \infty$ and every $\eps > 0$ there exist codes of rate $1 - \delta - \eps$ and constant alphabet (so $q = O_{\delta,\gamma,\eps}(1)$) and sub-logarithmic list sizes. Furthermore, our codes are accompanied by efficient (polynomial time) decoding algorithms. We stress that the fraction of insertions can be arbitrarily large (more than 100\%), and the rate is independent of this parameter. We also prove several tight bounds on the parameters of list-decodable insdel codes. In particular, we show that the alphabet size of insdel codes needs to be exponentially large in $\eps^{-1}$, where $\eps$ is the gap to capacity above. Our result even applies to settings where the unique-decoding capacity equals the list-decoding  capacity and when it does so, it shows that the alphabet size needs to be exponentially large in the gap to capacity. This is sharp contrast to the Hamming error model where alphabet size polynomial in $\eps^{-1}$ suffices for unique decoding. This lower bound also shows that the exponential dependence on the alphabet size in previous works that constructed insdel codes is actually necessary! 

Our result sheds light on the remarkable asymmetry between the impact of insertions and deletions from the point of view of error-correction: Whereas deletions cost in the rate of the code, insertion costs are borne by the adversary and not the code! Our results also highlight the dominance of the model of insertions and deletions over the Hamming model: A Hamming error is equal to one insertion and one deletion (at the same location). Thus the effect of $\delta$-fraction Hamming errors can be simulated by $\delta$-fraction of deletions and $\delta$-fraction of insertions --- but insdel codes can deal with much more insertions without loss in rate\fullOnly{ (though at the price of higher alphabet size)}. 

\end{abstract}
	
\setcounter{page}{0}
\thispagestyle{empty}

\newpage
\section{Introduction}

We study the complexity of ``insdel coding'', i.e., codes designed to recover from insertion and deletion of characters, under the model of ``list-decoding'', i.e., when the decoding algorithm is allowed to report a (short) list of potential codewords that is guaranteed to include the transmitted word if the number of errors is small enough.
Recent work by the first two authors and collaborators~\cite{haeupler2017synchronization} has shown major progress leading to tight, or nearly tight, bounds on central parameters of  codes (with efficient encoding and decoding algorithms as well) under the setting of {\em unique} decoding. However the list-decoding versions of these questions were not explored previously. Our work complements the previous studies by exploring list-decoding. In the process our results also reveal some striking features of the insdel coding problem that were not exposed by previous works. To explain some of this, we introduce our model and lay out some of the context below.

\subsection{Insdel Coding and List Decoding}

We use the phrase ``insdel coding'' to describe the study of codes that are aimed to recover from {\em insertions} and {\em deletions}. The principal question we ask is ``what is the {\em rate} of a code that can recover from $\gamma$ fraction insertions and $\delta$ fraction deletions over a sufficiently large alphabet?''. Once the answer to this question is determined we ask how small an alphabet suffices to achieve this  rate. We define the terms ``rate'', ``alphabet'', and ``recovery'' below. 

An insdel {\em encoder} over alphabet $\Sigma$ of block length $n$ is an injective function $E:\Sigma^k \to \Sigma^n$. The associated ``code'' is the image of the function $C$. The rate of a code is the ratio $k/n$. We say that an insdel code $C$ is $(\gamma,\delta,L(n))$-list-decodable if there exists a function $D:\Sigma^*\to 2^C$ such that $|D(w)| \leq L(n)$ for every $w \in \Sigma^*$ and 
for every codeword $x \in C$ and every word $w$ obtained from $x$ be $\delta\cdot n$ deletions of characters in $x$ followed by $\gamma \cdot n$ insertions, it is the case that $x \in D(w)$.
In other words the list-decoder $D$ outputs a list of at most $L(n)$ codewords that is guaranteed to include the transmitted word $x$ if the received word $w$ is obtained from $x$ by at most $\delta$-fraction deletions and $\gamma$-fraction insertions. Our primary quest in this paper is the largest rate $R$ for which there exists an alphabet $q \triangleq |\Sigma|$ and an
infinite family of insdel codes of rate at least $R$, that are $(\gamma,\delta,L(n))$-list-decodable. Of course we are interested in results where $L(n)$ is very slowly growing with $n$ (if at all). In all results below we get $L(n)$ which is polynomially large in terms of $n$.  
Furthermore, when a given rate is achievable we seek codes with efficient encoder and decoder (i.e., the functions $E$ and $D$ are polynomial time computable). Finally we also explore the dependence of the rate on the alphabet size (or vice versa).

\paragraph{Previous Work.} 
Insdel coding was first studied by Levenshtein~\cite{Levenshtein65} and since then many bounds and constructions for such codes have been given. 
With respect to unique decoding, Schulman and Zuckerman~\cite{schulman1999asymptotically} gave the first 
construction of efficient insdel codes over a constant alphabet with a (small) constant relative distance and a (small) constant rate in 1999.
Guruswami and Wang~\cite{guruswami2017deletion} gave the first efficient codes over fixed alphabets to correct a deletion fraction approaching $1$, as well as efficient binary codes to correct a small constant fraction of deletions with rate approaching $1$. A follow-up work gave new and improved codes with similar rate-distance tradeoffs which can be efficiently decoded from insertions and deletions~\cite{guruswami2016efficiently}.  Finally,~\cite{haeupler2017synchronization} gave codes that can correct $\delta$ fraction of synchronization errors with a rate approaching $1-\delta-\eps$.

A recent work by Wachter-Zeh~\cite{wachter2017list} considers insdel coding with respect to list decoding and provides Johnson-like upper-bounds for insertions and deletions, i.e., bounds on the list size in terms of minimum edit-distance of a given code.

Several other variants of the insdel coding problem have been studied in the previous work and are summarized by the following surveys~\cite{sloane2002single, mitzenmacher2009survey, mercier2010survey}.

\subsection{Our Results}

We now present our results on the rate and alphabet size of insdel coding under list-decoding. Two points of contrast that we use below are corresponding bounds in (1) the Hamming error setting for list-decoding and (2) the insdel coding setting with unique-decoding. 

\subsubsection{Rate Under List Decoding}

Our main theorem for list-decoding shows that, given $\gamma,\delta,\eps \geq 0$ there is a $q = q_{\eps,\gamma}$ and a slowly growing function $L = L_{\eps,\gamma}(n)$ such that there are $q$-ary insdel codes that achieve a rate of $1-\delta-\eps$ that are $(\gamma,\delta,L(n))$-list decodable. Furthermore the encoding and decoding are efficient! The formal statement of the main result is as follows.

\begin{theorem}\label{thm:InsDelListDecoding}
For every $0<\delta,\eps<1$ and $\gamma > 0$, there exist a family of list-decodable insdel codes that can protect against $\delta$-fraction of deletions and $\gamma$-fraction of insertions and achieves a rate of at least $1-\delta-\eps$ or more over an alphabet of size 
$O_{\gamma, \eps}\left(1\right)$
. These codes are list-decodable with lists of size $L_{\eps, \gamma}(n)= \exp\left(\log^* n\right)$, and have polynomial time encoding and decoding complexities.
\end{theorem}

The rate in the theorem above is immediately seen to be optimal even for $\gamma = 0$. In particular an adversary that deletes that last $\delta \cdot n$ symbols already guarantees an upper bound on the rate of $1 - \delta$. 

We now contrast the theorem above with the two contrasting settings listed earlier. Under unique decoding the best possible rate that can be achieved with $\delta$-fraction deletions and $\gamma$-fraction insertions is upper bounded by $1 - (\gamma + \delta)$. Matching constructions have been achieved, only recently,  by Haeupler and Shahrasbi~\cite{haeupler2017synchronization}.
In contrast our rate has no dependence on $\gamma$ and thus dominates the above result. 
The only dependence on $\gamma$ is in the alphabet size and list-size and we discuss the need for this dependence later below.

We now turn to the standard ``Hamming error'' setting: Here an adversary may change an arbitrary $\delta$-fraction of the codeword symbols. In this setting it is well-known that given any $\eps > 0$, there are constants $q = q(\eps)$ and $L = L(\eps)$ and an infinite family of $q$-ary codes of rate at least $1 - \delta - \eps$ that are list-decodable from $\delta$ fraction errors with list size at most $L$. In a breakthrough from the last decade, Guruswami and Rudra~\cite{guruswami2008explicit} showed explicit codes that achieve this with efficient algorithms. The state-of-the-art results in this field yield list size 
$L(n) = o(\log^{(r)} n)$ for any integer $r$ where $\log^{(r)}$ is the $r$th iterated logarithm and alphabet size $2^{\tilde{O}(\eps^{-2})}$~\cite{GuruswamiXing2017}, which are nearly optimal. 

The Hamming setting with $\delta$-fraction errors is clearly a weaker setting than the setting with $\delta$-fraction deletions and $\gamma \geq \delta$ fraction of insertions in that an adversary of the latter kind can simulate the former. (A Hamming error is a deletion followed by an insertion at the same location.) The insdel setting is thus stronger in two senses: it allows $\gamma > \delta$ and gives greater flexibility to the adversary in choosing locations of insertions and deletions. Yet our theorem shows that the stronger adversary can still be dealt with, with out qualitative changes in the rate. The only difference is in the dependence of $q$ and $L$ on $\gamma$, which we discuss next.

We briefly remark at this stage that while our result simultaneously ``dominates'' the results of Haeupler and Shahrasbi~\cite{haeupler2017synchronization} as well as Guruswami and Rudra~\cite{guruswami2008explicit} this happens because we use their results in our work. 
We elaborate further on this in Section~\ref{sec:InsDelListDecoding}. Indeed our first result (see Theorem~\ref{thm:IndexingBlackBox}) shows how we can obtain Theorem~\ref{thm:InsDelListDecoding} by using capacity achieving ``list-recoverable codes'' in combination with synchronization strings in a modular fashion.

\subsubsection{Rate versus Alphabet Size}

We now turn to understanding how large the alphabet size needs to be as a function of $\delta,\eps$ and $\gamma$. We consider two extreme cases, first with only deletions (i.e., $\gamma=0$ and then with only insertions (i.e., with $\delta = 0$). 

We start first with the insertion-only setting. We note  here that one cannot hope to find a constant rate family of codes that can protect $n$ symbols out of an alphabet of size $q$ against $(q-1)n$ many insertions or more. This is so since, with $(q-1)n$ insertions, one can turn any string $y\in[1..q]^n$ into the fixed sequence $1, 2, \cdots, q, 1, 2, \cdots, q, \cdots, 1, 2, \cdots, q$ by simply inserting $q-1$ many symbols around each symbol of $y$ to construct a $1, \cdots, q$ there. Hence, Theorem~\ref{thm:converse} only focuses on codes that protect $n$ rounds of communication over an alphabet of size $q$ against $\gamma n$ insertions for $\gamma  < q-1$.
 
\begin{theorem}\label{thm:converse}
Any list-decodable family of codes $\mathcal{C}$ that protects against $\gamma$ fraction of insertions for some $\gamma < q-1$ and guarantee polynomially-large list size in terms of block length cannot achieve a rate $R$ that is strictly larger than $1- \log_q(\gamma+1)-\gamma\left(\log_q\frac{\gamma+1}{\gamma} - \log_q \frac{q}{q-1}\right)$.
\end{theorem}

In particular the theorem asserts that if the code has rate $R = 1-\eps$, then its alphabet size must be exponentially large in $1/\eps$, namely, $q \geq (\gamma+1)^{1/\eps}$.

Next we turn to the deletion-only case. Here again we note that no constant rate $q$-ary codes can protect against $\delta \ge \frac{q-1}{q}$ fraction of deletions since such a large fraction of deletions may remove all but the least frequent symbol of codewords. Therefore, Theorem~\ref{thm:converseDeletion} below only concerns codes that protect against $\delta\le\frac{q-1}{q}$ fraction of deletions.
\begin{theorem}\label{thm:converseDeletion}
Any list-decodable family of insdel codes that protect against $\delta$-fraction of deletions (and no insertions) for some $0\le \delta < \frac{q-1}{q}$ that are lis-decodable with  polynomially-bounded list size has rate $R$ upper bounded as below:
\begin{itemize}
\item $R \leq f(\delta) \triangleq (1-\delta)\left(1-\log_q\frac{1}{1-\delta}\right)$ where $\delta = \frac{d}{q}$ for some integer $d$.
\item $
R \leq (1-q\delta')f\left(\frac{d}{q}\right)+
q\delta' f\left(\frac{d+1}{q}\right)$ where $\delta = \frac{d}{q}+\delta'$ for some integer $d$ and $0\le \delta' < \frac{1}{q}$.
\end{itemize}
\end{theorem}
In particular if $\delta = d/q$ for integer $d$ and rate is $1 - \delta - \eps$ then the theorem above asserts that $q \geq \left(\frac1{1-\delta}\right)^{\frac{1-\delta}{\eps}}$, or in other words $q$ must be exponentially large in $1/\eps$. Indeed such a statement is true for all $\delta$ as asserted in the corollary below. 

\begin{corollary}\label{cor:alphabetSize}
There exists a function $f:(0,1) \to (0,1)$ such that 
Any family of insdel codes that protects against $\delta$-fraction of deletions with polynomially bounded list sizes and has rate $1 - \delta - \eps$ must have alphabet size $q \geq \exp\left(\frac{f(\delta)}{\eps}\right)$.
\end{corollary}

\paragraph{Implications for Unique Decoding.}

Even though the main thrust of this paper is list-decoding, Corollary~\ref{cor:alphabetSize} size also has implications for unique-decoding. (This turns out to be a consequence of the fact that the list-decoding radius for deletions-only equals the unique-decoding radius for the same fraction of deletions.) We start by recalling the main result of Haeupler and Shahrasbi~\cite{haeupler2017synchronization}: Given any $\alpha,\eps > 0$ there exists a code of rate $1-\alpha-\eps$ over an alphabet of size $q = \exp(1/\eps)$ that {\em uniquely} decodes from any $\alpha$-fraction synchronization errors, i.e., from   $\gamma$-fraction insertions and $\delta$-fraction deletions for any pair $0 \leq \gamma,\delta$ satisfying $\gamma + \delta \leq \alpha$. Furthermore, this is the best possible rate one can achieve for $\alpha$-fraction synchronization error. (See Appendix~\ref{app:UniqueDecodingRate} for a more detailed description with proof.)

Till now this exponential dependence of the alphabet size on $\eps$ was unexplained. This is also in sharp contrast to the Hamming error setting, where codes are known to get $\eps$ close to unique decoding capacity (half the ``Singleton bound'' on the distance of code) with alphabets of size polynomial in $1/\eps$. Indeed given this contrast one may be tempted to believe that the exponential growth is a weakness of the ``synchronization string'' approach of Haeupler and Shahrasbi~\cite{haeupler2017synchronization}. But Corollary~\ref{cor:alphabetSize} actually shows that an exponential bound is necessary. We state this result for completeness due to the even though it is immediate from the Corollary above, to stress its importance in understanding the nature of synchronization errors.

\begin{corollary}\label{cor:unique-alphabetSize}
There exists a function $f:(0,1) \to (0,1)$ such that for every $\alpha,\eps>0$ every
family of insdel codes of rate $1-\alpha-\eps$ that protects against $\alpha$-fraction of synchronization errors with unique decoding must have alphabet size $q \geq \exp\left(\frac{f(\delta)}{\eps}\right)$.
\end{corollary}

Corollary~\ref{cor:unique-alphabetSize} follows immediately from Corollary~\ref{cor:alphabetSize} by setting $\delta = \alpha$ and $\gamma=0$ (so we get to the zero insertion case) and noticing that a unique-decoding insdel code for $\alpha$-fraction synchronization error is also a list-decoding insdel code for $\delta$-fractions of deletions (and no insertions). The alphabet size lower bound for the latter is also thus a alphabet size lower bound for the former.

We note that it would be desirable to prove a stronger lower bound on alphabet size of every code that unique decodes with $\delta$-fraction deletions, $\gamma$-fraction insertions and has rate $1 - \delta - \gamma - \eps$ but we do not have such a result yet.

\subsubsection{Analysis of Random Codes}
Finally, in Section~\ref{sec:RandomCodeLowerBounds}\shortOnly{ and Appendix~\ref{sec:RandomCodeLowerBoundsCont}}, we provide an analysis of random codes and compute the rates they can achieve while maintaining list-decodability against insertions and deletions. Such rates are essentially lower-bounds for the capacity of insertion and deletion channels and can be compared against the upper-bounds provided in Section~\ref{sec:UpperBounds}. 

Theorem~\ref{thm:achievabilityViaRandomCodesDeletion} shows that the family of random codes over an alphabet of size $q$ can, with high probability, protect against $\delta$-fraction of deletions for any $\delta < 1-1/q$ up to a rate of 
 $1-(1-\delta)\log_q\frac{1}{1-\delta} - \delta \log_q\frac{1}{\delta}-\delta\log_q(q-1) = 1-H_q(\delta)$ 
using list decoding with super-constant list sizes in terms of their block length where $H_q$ represents the $q$-ary entropy function.
 \begin{theorem}\label{thm:achievabilityViaRandomCodesDeletion}
For any alphabet of size $q$ and any $0 \le \delta < \frac{q-1}{q}$, the family of random codes with rate 
$R < 1-(1-\delta)\log_q\frac{1}{1-\delta} - \delta \log_q\frac{1}{\delta}-\delta\log_q(q-1)-\frac{1-\delta}{l+1}$ 
is list-decodable with list size of $l$ from any $\delta$ fraction of deletions with high probability. Further, the family of random deletion-codes with rate $R > 1-(1-\delta)\log_q\frac{1}{1-\delta} - \delta \log_q\frac{1}{\delta}-\delta\log_q(q-1)$ is not list-decodable with high probability.
\end{theorem}

Further, Theorem~\ref{thm:achievabilityViaRandomCodes} shows that the family of random block codes over an alphabet of size $q$ can, with high probability, protect against $\gamma$ fraction of insertions for any $\gamma < q-1$ up to a rate of 
$1  - \log_q(\gamma +1) - \gamma\log_q\frac{\gamma+1}{\gamma}$ using list decoding with super-constant list sizes in terms of block length. 

\begin{theorem}\label{thm:achievabilityViaRandomCodes}
For any alphabet of size $q$ and any $\gamma < q-1$, the family of random codes with rate 
$R < 1-\log_q(\gamma +1) - \gamma\log_q\frac{\gamma+1}{\gamma} - \frac{\gamma +1}{l+1}$ 
is list-decodable with a list size of $l$ from any $\gamma n$ insertions with high probability.
\end{theorem}
\section{Definitions and Preliminaries}\label{sec:prelim}
\subsection{Synchronization Strings}
In this section, we briefly recapitulate synchronization strings, introduced by Haeupler and Shahrasbi~\cite{haeupler2017synchronization} and further studied in~\cite{haeupler2017synchronization2:ARXIV, haeupler2017synchronization3}. We will review important definitions and techniques from~\cite{haeupler2017synchronization} that will be of use throughout this paper.

Synchronization strings are recently introduced mathematical objects that turn out to be useful tools to overcome synchronization errors, i.e., symbol insertion and symbol deletion errors. The general idea employed in~\cite{haeupler2017synchronization, haeupler2017synchronization2:ARXIV} to obtain resilience against synchronization errors in various communication setups is  \emph{indexing} each symbol of the communication with symbols of a synchronization string and then \emph{guessing} the actual position of received symbols on the other side using indices. \cite{haeupler2017synchronization} provides a variety of different guessing strategies that guarantee a large number of correct guesses and then overcome the incorrect guesses by utilizing classic error correcting codes. As a matter of fact, synchronization strings essentially translate synchronization errors into ordinary Hamming type errors which are strictly easier to handle. We now proceed to review some of the above-mentioned definitions and techniques more formally.

Suppose that two parties are communicating over a channel that suffers from $\alpha$-fraction of insertions and deletions and one of the parties sends a pre-shared string $S$ of length $n$ to the other one. A distorted version of $S$ will arrive at the receiving end that we denote by $S'$. A symbol $S[i]$ is called to be a \emph{successfully transmitted} symbol if it is not removed by the adversary. A \emph{decoding algorithm} in the receiving side is an algorithm that, for any received symbol, guesses its actual position in $S$ by either returning a number in $[1..n]$ or $\top$ which means the algorithm is not able to guess the index. For such a decoding algorithm, a successfully transmitted symbol whose index is not guessed correctly by the decoding algorithm is called a \emph{misdecoding}.

Haeupler and Shahrasbi~\cite{haeupler2017synchronization} introduce synchronization strings and find several decoding algorithms for them providing strong misdecoding guarantees and then design insertion-deletion codes based on those decoding algorithms. As details of those algorithms are not relevant to this paper we avoid further discussion of those techniques. To conclude this section we introduce $\eps$-synchronization strings and an important property of them.

\begin{definition}[$\eps$-synchronization string]\label{def:synCode}
String $S \in \Sigma^n$ is an $\eps$-synchronization string if for every $1 \leq i < j < k \leq n + 1$ we have that $ED\left(S[i, j),S[j, k)\right) > (1-\eps) (k-i)$. 
\end{definition}

The key idea in the construction of insdel codes in~\cite{haeupler2017synchronization} is to index an error correcting code with a synchronization string. Here we provide a formal definition of indexing operation.
\begin{definition}[Indexing]
The operation of indexing code $\mathcal{C}$ with block length $n$ and String $S$ of length $n$ is to simply replace each codeword $w_1,w_2,\cdots, w_n$ with $(w_1, s_1), (w_2, s_2), \cdots, (w_n, s_n)$. Clearly, this operations expands the alphabet of the code.
\end{definition}

It is shown in~\cite{haeupler2017synchronization, haeupler2017synchronization3} that $\eps$-synchronization strings exist over alphabets of sizes polynomially large in terms of $\eps^{-1}$ and can be efficiently constructed.
An important property of $\eps$-synchronization strings discussed in~\cite{haeupler2017synchronization} is the self matching property defined as follows.

\begin{definition}[$\eps$-self-matching property]
String $S$ satisfies $\eps$-self-matching property if for any two sequences of indices $1\leq a_1 < a_2 < \cdots < a_k\le |S|$ and $1\leq b_1 < b_2 < \cdots < b_k\le |S|$ that satisfy $S[a_i]=S[b_i]$ and $a_i \neq b_i$, $k$ is not larger than $\eps|S|$.
\end{definition}

In the end, we review the following theorem from~\cite{haeupler2017synchronization} that shows the close connection between synchronization string property and the self-matching property.

\begin{theorem}[Theorem 6.4 from~\cite{haeupler2017synchronization}]\label{thm:selfMatchingProperty}
If $S$ is an $\eps$-synchronization string, then all substrings of $S$ satisfy $\eps$-self-matching property.
\end{theorem}

\subsection{List Recoverable Codes}
A code $\mathcal{C}$ given by the encoding function $\mathcal{E}:\Sigma^{nr}\rightarrow\Sigma^n$ is called to be $(\alpha, l, L)$-list recoverable if for any collection of $n$ sets $S_1, S_2, \cdots, S_n\subset\Sigma$ of size $l$ or less, there are at most $L$ codewords for which more than $\alpha n$ elements appear in the list that corresponds to their position, i.e.,  $$\left|\left\{x \in \mathcal{C}\mid \left|\left\{i \in [n] \mid x_i \in S_i\right\}\right| \geq \alpha n \right\}\right| \leq L.$$

The study of list-recoverable codes was inspired by Guruswami and Sudan's list-decoder for Reed-Solomon codes~\cite{guruswami1998improved}. Since then, list-recoverable codes have became a very useful tool in coding theory~\cite{guruswami2001expander, guruswami2002near, guruswami2003linear, guruswami2004linear} and there have been a variety of constructions provided for them by several works~\cite{guruswami2008explicit, guruswami2011optimal, guruswami2013list, kopparty2015list, hemenway2015linear, GuruswamiXing2017, hemenway2017local}.
In this paper, we will make use of the following capacity-approaching polynomial-time list-recoverable codes given by Hemenway, Ron-Zewi, and Wootters~\cite{hemenway2017local} that is obtained by altering the approach of Guruswami and Xing~\cite{guruswami2013list}.

\begin{theorem}[Hemenway et. al.~{\cite[Theorem A.7]{hemenway2017local}}]\label{thm:PolyTimeListRecoverable}
Let $q$ be an even power of a prime, and choose $l, \epsilon > 0$, so that $q \geq \epsilon^{-2}$. Choose $\rho \in (0, 1)$. There is an $m_{min} = O(l \log_q(l/\epsilon)/\epsilon^2)$ so that the following holds for all $m \geq m_{min}$. For infinitely many $n$ (all $n$ of the form $q^{e/2}(\sqrt q - 1)$ for any integer $e$), there is a deterministic polynomial-time construction of an $F_q$-linear code $C:\mathbb{F}^{\rho n}_{qm} \rightarrow \mathbb{F}^n_{qm}$ of rate $\rho$ and relative distance $1-\rho - O(\epsilon)$ that is $(1 - \rho - \epsilon, l, L)$-list-recoverable in time $\poly(n, L)$, returning a list that is
contained in a subspace over $\mathbb{F}_q$ of dimension at most 
$\left(\frac{l}{\epsilon}\right)^{2\log^*(mn)}$.
\end{theorem}
 

\section{List Decoding for Insertions and Deletions}\label{sec:InsDelListDecoding}

In this section, we prove Theorem~\ref{thm:InsDelListDecoding} by constructing a list-decodable code of rate $1-\delta-\eps$ that provides resilience against $0<\delta<1$ fraction of deletions and $\gamma$ fraction of insertions over a constant-sized alphabet.
Our construction heavily relies on the following theorem that, in the same fashion as~\cite{haeupler2017synchronization}, uses the technique of indexing an error correcting code with a synchronization string to convert a given list-recoverable code into an insertion-deletion code.

\begin{theorem}\label{thm:IndexingBlackBox}
Let $\mathcal{C}:\Sigma^{nR}\rightarrow\Sigma^n$ be a $(\alpha, l, L)$-list recoverable code with rate $R$, encoding complexity $T_{Enc}$ and decoding complexity complexity $T_{Dec}$. For any $\eps>0$ and $\gamma \leq \frac{l\eps}{2}-1$, by indexing $\mathcal{C}$ with an $\frac{\eps^2}{4(1+\gamma)}$-synchronization string, one can obtain an insertion-deletion code $\mathcal{C}': \Sigma^{nr}\rightarrow[\Sigma\times\Gamma]^n$ that corrects from $\delta < 1-\alpha-\eps$ fraction of deletions and $\gamma$ fraction of insertions where $|\Gamma| = \left({\eps^2}/{(1+\gamma)}\right)^{-O(1)}$. $\mathcal{C}'$ is encodable and decodable in $O(T_{Enc}+n)$ and $O(T_{Dec}+n^2(1+\gamma^2)/\eps)$ time respectively.
\end{theorem}

We take two major steps to prove Theorem~\ref{thm:IndexingBlackBox}.
In the first step (Theorem~\ref{thm:InsDelListDecodingStep1}), we use the synchronization string indexing technique from~\cite{haeupler2017synchronization} and show that by indexing the symbols that are conveyed through an insertion-deletion channel with symbols of a synchronization string, the receiver can make \emph{lists} of candidates for any position of the sent string such that $1-\delta-\eps$ fraction of lists are guaranteed to contain the actual symbol sent in the corresponding step and the length of the lists is guaranteed to be smaller than some constant $O_{\gamma, \eps}(1)$. 

In the second step, we use list-recoverable codes on top of the indexing scheme to obtain a list decoding using lists of candidates for each position produced by the former step.

We start by the following lemma that directly implies the first step stated in Theorem~\ref{thm:InsDelListDecodingStep1}.

\begin{lemma}\label{lem:globalDecodingListSize}
Assume that a sequence of $n$ symbols denoted by $x_1x_2\cdots x_n$ is indexed with an $\eps$-synchronization string and is communicated through a channel that suffers from up to $\delta n$ deletions for some $0 \le \delta < 1$ and $\gamma n$ insertions. Then, on the receiving end, it is possible to obtain $n$ lists $A_1, \cdots, A_n$ such that, for any desired integer $K$, for at least $n\cdot\left(1-\delta - \frac{1+\delta}{K} - K\cdot \eps\right)$ of them, $x_i\in A_i$. All lists contain up to $K$ elements and the average list size is at most $1+\gamma$. These lists can be computed in $O\left(K(1+\gamma) n^2\right)$ time.
\end{lemma}
\begin{proof}
The decoding algorithm we propose to obtain the lists that satisfy the guarantee promised in the statement is the global algorithm introduced in Theorem 6.14 of Haeupler and Shahrasbi~\cite{haeupler2017synchronization}.


Let $S$ be the $\eps$-synchronization string used for indexing and $S'$ be the index portion of the received string on the other end. Note that $S$ is pre-shared between the sender and the receiver. The decoding algorithm starts by finding a longest common substring $M_1$ between $S$ and $S'$ and adding the position of any matched element from $S'$ to the list that corresponds to its respective match from side $S$.
 Then, it removes every symbol that have been matched from $S'$ and repeats the previous step by finding another longest common subsequence $M_2$ between $S$ and the remaining elements of $S'$. This procedure is repeated $K$ times to obtain $M_1, \cdots, M_K$. This way, lists $A_i$ are formed by including every element in $S'$ that is matched to $S[i]$ in any of $M_1, \cdots, M_K$.

$A_i$ contains the actual element that corresponds to $S[i]$, denoted by $S'[j]$, if and only if $S[i]$ is successfully transmitted (i.e., not removed by the adversary), appears in one of $M_k$s, and matches to $S[i]$ in $M_k$. Hence, there are three scenarios under which $A_i$ does not contain its corresponding element $S[i]$. 
\begin{enumerate}
\item $S[i]$ gets deleted by the adversary. 
\item $S[i]$ is successfully transmitted but, as $S'[j]$ on the other side, it does not appear on any of $M_k$s.
\item $S[i]$ is successfully transmitted and, as $S'[j]$ on the other side, it appears in some $M_k$ although it is matched to another element of $S$.
\end{enumerate}

The first case happens for at most $\delta n$ elements as adversary is allowed to delete up to $\delta n$ many elements. 

To analyze the second case, note that sizes of $M_k$s descend as $k$ grows since we pick the longest common subsequence in each step. If by the end of this procedure $p$ successfully transmitted symbols are still not matched in any of the matchings, they form a common subsequence of size $p$ between $S$ and the remainder of $S'$. This leads to the fact that $|M_1| + \cdots + |M_K| \ge K\cdot p$. As $|M_1| + \cdots + |M_K|$ cannot exceed $|S'|$, we have $p \le |S'|/K$. This bounds above the number of symbols falling into the second category by $|S'|/K$.

Finally, as for the third case, we draw the reader's attention to the fact that each successfully transmitted $S[i]$ which arrives at the other end as $S'[j]$ and mistakenly gets matched to another element of $S$ like $S[k]$ in $M_t$, implies that $S[i] = S[k]$. We call the pair $(i, k)$ \emph{a pair of similar elements in $S$ implied by $M_t$}.
Note that there is an actual monotone matching from $S$ to $S'$ that corresponds to adversary's actions like $M'$. As $M_1$ and $M'$ are both monotone, the set of similar pairs in $S$ implied by $M_t$ is a self-matching in $S$. As stated in Theorem~\ref{thm:selfMatchingProperty}, the number of such pairs cannot exceed $n\eps$. Therefore, there can be at most $n\eps$ successfully transmitted symbols that get mistakenly matched in $M_1$. Same argument holds for the rest of $M_k$s, hence, the number of elements falling into the third category is at most $nK\eps$.

Summing up all above-mentioned bounds gives that the number of bad lists can be bounded above by\fullOnly{ the following.
$$n\delta + \frac{|S'|}{K} + nK\eps \le n\left(\delta + \frac{1+\gamma}{K} + K\eps\right)$$}
\shortOnly{$n\delta + \frac{|S'|}{K} + nK\eps \le n\left(\delta + \frac{1+\gamma}{K} + K\eps\right)$.}
This proves the list quality guarantee. As proposed decoding algorithm computes longest common substring $K$ many times between two strings of length $n$ and $(1+\gamma)n$ or less, it will run in $O(K(1+\gamma)\cdot n^2)$ time.
\end{proof}

\begin{theorem}\label{thm:InsDelListDecodingStep1}
Suppose that $n$ symbols denoted by $x_1, x_2, \cdots, x_n$ are being communicated through a channel suffering from up to $\delta n$ deletions for some $0 \le \delta < 1$ and $\gamma n$ insertions for some constant $\gamma\geq0$. If one indexes these symbols with an $\eps' = \frac{\eps^2}{4(1+\gamma)}$-synchronization string, then, on the receiving end, it is possible to obtain $n$ lists $A_1, \cdots, A_n$ such that for at least $n\cdot\left(1-\delta - \eps\right)$ of them $x_i\in A_i$. Each list contains up to $2(1+\gamma)/\eps$ elements and the average list size is at most $1+\gamma$. These lists can be computed in $O\left(n ^2(1+\gamma)^2/\eps\right)$ time.
\end{theorem}
\begin{proof}
Using an $\eps' = \frac{\eps^2}{4(1+\gamma)}$-synchronization string in the statement of Lemma~\ref{lem:globalDecodingListSize} and choosing $K = \frac{2(1+\gamma)}{\eps}$ directly gives that the runtime is $O\left(n ^2(1+\gamma)^2/\eps\right)$ and list hit ration is at least 
\shortOnly{$n\cdot\left(1-\delta - \frac{1+\gamma}{K} - K\cdot \eps'\right) = n\cdot(1-\delta-\eps/2-\eps/2) = n\cdot(1-\delta-\eps)$}
\fullOnly{$$n\cdot\left(1-\delta - \frac{1+\gamma}{K} - K\cdot \eps'\right) = n\cdot(1-\delta-\eps/2-\eps/2) = n\cdot(1-\delta-\eps)$$}
\end{proof}

Theorem~\ref{thm:InsDelListDecodingStep1} facilitates the conversion of list-recoverable error correcting codes into list-decodable insertion-deletion codes as stated in Theorem~\ref{thm:IndexingBlackBox}.

\begin{proof}[{\bf Proof of Theorem~\ref{thm:IndexingBlackBox}}]
To prove this, we simply index code $\mathcal{C}$, entry by entry, with an $\eps'=\frac{\eps^2}{4(1+\gamma)}$ synchronization string. In the decoding procedure, according to Theorem~\ref{thm:InsDelListDecodingStep1}, the receiver can use the index portion of the received symbol to maintain lists of up to $2(1+\gamma)/\eps \leq l$ candidates for each position of the sent codeword of $\mathcal{C}$ so that $1-\delta-\eps > \alpha$ fraction of those contain the actual corresponding sent message. Having such lists, the receiver can use the decoding function of $\mathcal{C}$ to obtain an $L$-list-decoding for $\mathcal{C}'$. Finally, the alphabet size and encoding complexity follow from the fact that synchronization strings over alphabets of size ${\eps'}^{-O(1)}$ can be constructed in linear time~\cite{haeupler2017synchronization, haeupler2017synchronization3}.
\end{proof}

One can use any list-recoverable error correcting code to obtain insertion-deletion codes according to Theorem~\ref{thm:IndexingBlackBox}. In particular, using the efficient capacity-approaching list-recoverable code introduced by 
Hemenway, Ron-Zewi, and Wootters~\cite{hemenway2017local}
, one obtains the insertion-deletion codes as described in Theorem~\ref{thm:InsDelListDecoding}.
\begin{proof}[{\bf Proof of Theorem~\ref{thm:InsDelListDecoding}}]
By setting parameters $\rho=1-\delta-\frac{\eps}{2}$, $l=\frac{2(\gamma+1)}{\eps}$, and $\epsilon=\frac{\eps}{4}$ in Theorem~\ref{thm:PolyTimeListRecoverable}, one can obtain a family of codes $\mathcal{C}$ that achieves rate $\rho=1-\delta-\frac{\eps}{2}$ and is $(\alpha, l, L)$-recoverable in polynomial time for $\alpha=1-\delta-\eps/4$ and some $L=\exp\left(\log^*n\right)$ (by treating $\gamma$ and $\eps$ as constants). 
Such family of codes can be found over an alphabet $\Sigma_{\mathcal{C}}$ of size $q=O_{\gamma, \eps}(1)$ or infinitely many integer numbers larger than $q$.

Plugging this family of codes into the indexing scheme from Theorem~\ref{thm:IndexingBlackBox} by choosing the parameter $\eps'=\frac{\eps}{4}$, one obtains a family of codes that can recover from
$1-\alpha-\eps' = 1-(1-\delta-\eps/4)-\eps/4 = \delta$ fraction of deletions and $\gamma$-fraction of insertions and achieves a rate of 
$$\frac{1-\delta-\eps/2}{1+\frac{\log|\Sigma_S|}{\log |\Sigma_{\mathcal{C}}|}}$$
which, by taking $|\Sigma_{\mathcal{C}}|$ large enough in terms of $\eps$, is larger than $1-\delta-\eps$. As $\mathcal{C}$ is encodable and decodable in polynomial time, the encoding and decoding complexities of the indexed code will be polynomial as well.
\end{proof}

\begin{remark}
We remark that by using capacity-approaching near-linear-time list-recoverable code introduced in Theorem 7.1 of Hemenway, Ron-Zewi, and Wootters~\cite{hemenway2017local} in the framework of Theorem~\ref{thm:IndexingBlackBox}, one can obtain similar list-decodable insertion-deletion codes as in Theorem~\ref{thm:InsDelListDecoding} with a randomized quadratic time decoding. Further, one can use the efficient list-recoverable in the recent work of Guruswami and Xing~\cite{GuruswamiXing2017} to obtain same result as in Theorem~\ref{thm:InsDelListDecoding} except with polylogarithmic list sizes.
\end{remark}

\section{Upper Bounds on the Rate of List-Decodable Synchronization Codes}\label{sec:UpperBounds}
\subsection{Deletion Codes (Theorem~\ref{thm:converseDeletion})}
\begin{proof}[Proof of Theorem~\ref{thm:converseDeletion}]
To prove this claim, we propose a strategy for the adversary which can reduce the number of strings that may possibly arrive at the receiving side to a number small enough that implies the claimed upper bound for the rate. 

We start by proving the theorem for the case where $\delta q$ is integer. For an arbitrary code $\mathcal{C}$, upon transmission of any codeword, the adversary can remove all occurrences of $\delta q$ least frequent symbols as the total number of appearances of such symbols does not exceed $\delta n$. In case there are more deletions left, adversary may choose to remove arbitrary symbols among the remaining ones. This way, the received string would be a string of $n(1-\delta)$ symbols consisted of only $q-q\delta$ many distinct symbols. Therefore, one can bound above the size of the ensemble of strings that can possibly be received by the\fullOnly{ following.
$$|\mathcal{E}| \le {q \choose q(1-\delta)}\left[q(1-\delta)\right]^{n(1-\delta)}$$}\shortOnly{$|\mathcal{E}| \le {q \choose q(1-\delta)}\left[q(1-\delta)\right]^{n(1-\delta)}$.}
As the best rate any code can get is at most $\frac{\log |\mathcal{E}|}{n\log q}$, the following would be an upper bound for the best rate one might hope for.
\shortOnly{$$R \le \frac{\log |\mathcal{E}|}{n\log q} = \frac{\log {q \choose q(1-\delta)} + n(1-\delta)(\log (q(1-\delta)))}{n\log q} = (1-\delta)\left(1-\log_q \frac{1}{1-\delta}\right)+ o(1)$$}\fullOnly{
\begin{eqnarray*}
R &\le& \frac{\log |\mathcal{E}|}{n\log q}\\
 &=& \frac{\log {q \choose q(1-\delta)} + n(1-\delta)(\log (q(1-\delta)))}{n\log q}\\
&=& (1-\delta)\left(1-\log_q \frac{1}{1-\delta}\right)+ o(1)
\end{eqnarray*}}
This shows that for the case where $q\delta$ is an integer number, there are no family of codes that achieve a rate that is strictly larger than 
$(1-\delta)\left(1-\log_q \frac{1}{1-\delta}\right)$.

We now proceed to the general case where $\delta = d/q + \delta'$ for some integer $d$ and $0\le \delta' < \frac{1}{q}$. We closely follow the idea that we utilized for the former case. The adversary can partition $n$ sent symbols into two parts of size $nq\delta'$ and $n(1-q\delta')$, and then, similar to the former case, removes the $d+1$ least frequent symbols from the first part by performing $\frac{d+1}{q}\cdot nq\delta'$ deletions and $d$ least frequent symbols from the second one by performing $\frac{d}{q}\cdot n(1-q\delta')$ ones. This is possible because 
$\frac{d+1}{q}\cdot nq\delta' + \frac{d}{q}\cdot n(1-q\delta') = n\delta$.
Doing so, the string received after deletions would contain up to $q-d-1$ distinct symbols in its first $nq\delta'\left(1-\frac{d+1}{q}\right)$ positions and up to $q-d$ distinct symbols in the other $n(1-q\delta')\left(1-\frac{d}{q}\right)$ positions.
Therefore, the size of the ensemble of strings that can be received is bounded above as follows.
$$|\mathcal{E}| \le 
{q \choose q-d-1}\left[q-d-1\right]^{nq\delta'\left(1-\frac{d+1}{q}\right)}\cdot
{q \choose q-d}\left[q-d\right]^{n(1-q\delta')\left(1-\frac{d}{q}\right)}
$$
This bounds above the rate of $\mathcal{C}$ as follows.
\begin{eqnarray*}
R &\le& \frac{\log |\mathcal{E}|}{n\log q}\\
 &=& \frac{\log {q \choose q-d-1}  +
 nq\delta'\left(1-\frac{d+1}{q}\right)\log (q-d-1)
+ \log {q \choose q-d} +
 n(1-q\delta')\left(1-\frac{d}{q}\right)\log (q-d)
 }{n\log q}\\
&=& q\delta'\left[(1-(d+1)/q)\left(1-\log_q\frac{1}{1-(d+1)/q}\right)\right]
+
(1-q\delta')\left[(1-d/q)\left(1-\log_q\frac{1}{1-d/q}\right)\right]\\&&+ o(1)
\end{eqnarray*}
\end{proof}

\subsection{Insertion Codes (Theorem~\ref{thm:converse})}
Before providing the proof of Theorem~\ref{thm:converse}, we first point out that any $q-1$ insertions can be essentially used as a single erasure. As a matter of fact, by inserting $q-1$ symbols around the first symbol adversary can make a $1, 2, \cdots, q$ substring around first symbol and therefore, essentially, make the receiver unable to gain any information about it. In fact, with $\gamma n$ insertions, the adversary can repeat this procedure around any $\left\lfloor\frac{\gamma n}{q-1}\right\rfloor$ symbols he wishes. This basically gives that, with $\gamma n$ insertions, adversary can \emph{erase} $\left\lfloor\frac{\gamma n}{q-1}\right\rfloor$ many symbols. Thus, one cannot hope for finding list-decodable codes with rate $1-\frac{\gamma}{q-1}$ or more protecting against $\gamma n$ insertions. 

\begin{proof}[\bf Proof of Theorem~\ref{thm:converse}]
To prove this, consider a code $\mathcal{C}$ with rate 
$R\ge 1- \log_q(\gamma+1)-\gamma\left(\log_q\frac{\gamma+1}{\gamma} - \log_q \frac{q}{q-1}\right)+\eps$ 
for some $\eps > 0$. We will show that there exist $c_0^n$ many codewords in $C$ that can be turned into one specific string $z\in [1..q]^{n(\gamma+1)}$ with $\gamma n$ insertions for some constant $c_0 > 1$ that merely depends on $q$ and $\eps$.

First, the lower bound assumed for the rate implies that
\begin{equation}\label{eqn:CodeSize}
|C| = q^{nR} \ge q^{n\left(1-\log_q(\gamma+1)-\gamma\left(\log_q\frac{\gamma+1}{\gamma} - \log_q \frac{q}{q-1}\right)+\eps\right)}.
\end{equation}

Let $Z$ be a random string of length $(\gamma+1)n$ over the alphabet $[1..q]$. We compute the expected number of codewords of $\mathcal{C}$ that are subsequences of $Z$ denoted by $X$.

\begin{eqnarray}
\mathbb{E}[X]
 &=& \sum_{y\in\mathcal{C}}\Pr\{y\text{ is a subsequence of }Z\}\nonumber\\
&=& \sum_{y\in\mathcal{C}}\sum_{1\le a_1<a_2<\cdots <a_n \le n(\gamma +1)} \frac{1}{q^n}\left(1-\frac{1}{q}\right)^{a_n-n}\label{eqn:ExpectationComputationStep1}\\
&=& |\mathcal{C}|\left(q-1\right)^{-n}\sum_{l=n}^{n(1+\gamma)} {l \choose n} \left(\frac{q-1}{q}\right)^l\nonumber\\
&\le& |\mathcal{C}|\left(q-1\right)^{-n}n\gamma{n(1+\gamma) \choose n} \left(\frac{q-1}{q}\right)^{n(1+\gamma)}\label{eqn:ExpectationComputationStep2}\\
&=& n\gamma|\mathcal{C}| (q-1)^{n\gamma}q^{-n(1+\gamma)}2^{n(1+\gamma)H\left(\frac{1}{1+\gamma}\right)+o(n)}\nonumber\\
&=& n\gamma|\mathcal{C}| q^{n\left(
\gamma \log_q(q-1)
-1-\gamma
+\log_q (1+\gamma)
+\gamma\log_q \frac{1+\gamma}{\gamma}
\right)+o(1)}\nonumber\\
&=&q^{n\eps+o(n)}\label{eqn:ExpectationComputationStep3}
\end{eqnarray}
Step~\eqref{eqn:ExpectationComputationStep1} is obtained by conditioning the probability of $y$ being a subsequence of $Z$ over the leftmost occurrence of $y$ in $Z$ indicated by $a_1, a_2, \cdots, a_n$ as indices of $Z$ where the leftmost occurrence of $y$ is located. In that event, $Z_{a_i}$ has to be similar to $y_i$ and $y_i$ cannot appear in $Z[y_{i-1}+1, y_{i}-1]$. Therefore, the probability of this event is $\left(\frac{1}{q}\right)^n\left(1-\frac{1}{q}\right)^{a_n-n}$. 
To verify Step~\eqref{eqn:ExpectationComputationStep2}, we show that the summation in previous step takes its largest value when $l=n(1+\gamma)$ and bound the summation above by $n\gamma$ times that term. To see that ${l \choose n} \left(\frac{q-1}{q}\right)^l$ is maximized for $l=n(1+\gamma)$ in $n \le l \le n(1+\gamma)$ it suffices to show that the ratio of consecutive terms is larger than one for $l \le n(1+\gamma)$:
$$\frac{{l \choose n} \left(\frac{q-1}{q}\right)^{l}}{{l-1 \choose n} \left(\frac{q-1}{q}\right)^{l-1}} = \frac{l}{l-n} \cdot\frac{q-1}{q}=\frac{1-\frac{1}{q}}{1-\frac{n}{l}} \ge 1$$
The last inequality follows from the fact that $l\le n(\gamma+1)\le nq \Rightarrow \frac{1}{q} < \frac{n}{l}$.

Finally, by \eqref{eqn:ExpectationComputationStep3}, there exists some $z\in[1..q]^{(a+1)n}$ to which at least $q^{\eps n+o(n)}$, i.e., exponentially many codewords of $\mathcal{C}$ are subsequences. Therefore, polynomial-sized list decoding for received message $z$ is impossible and proof is complete.
\end{proof}

\global\def\AlphabetSizeSBGapDependenceSection{
\section{Dependence of the Alphabet Size on the Gap from the Singleton Bound (Corollary~\ref{cor:alphabetSize})}\label{sec:CapacityGap}

Before providing the proof of Corollary~\ref{cor:alphabetSize}, we start with an informal argument that shows the necessity of exponential dependence of alphabet size on $\eps^{-1}$. Consider a communication of $n$ symbols out of an alphabet of size $q$ where adversary is allowed to delete $\delta<\frac{1}{2}$ fraction of symbols. Note that one can map symbols of the alphabet to binary strings of length $\log q$. Let adversary act as follows. He looks at the first $2n\delta$ symbols and among symbols whose binary representation start with zero and symbols whose binary representation start with one chooses the least frequent ones and removes them. This can be done by up to $n\delta$ symbol deletions and he can use the remainder of deletions arbitrarily. This way, the receiver receives $n(1-\delta)$ symbols where only $q/2$ distinct symbols might appear in the first $n\delta$ ones. This means that receiver can essentially get $n(1-\delta)\log q - n\delta + 1$ bits of information which implies a rate upper bound of $1 - \delta - \frac{1}{\log q}$. Therefore, to get a rate of $1-\delta-\eps$, alphabet size has to satisfy the following.
$$\frac{1}{\log q} < \eps \Rightarrow q > 2^{\eps^{-1}}.$$

\begin{proof}[\bf Proof of Corollary~\ref{cor:alphabetSize}]
According to Theorem~\ref{thm:converseDeletion}, in order to obtain a family of codes of rate $1-\delta-\eps$, the following condition should hold.
\begin{eqnarray*}
1-\delta-\eps &\le& q\delta'\left[(1-(d+1)/q)\left(1-\log_q\frac{1}{1-(d+1)/q}\right)\right]\\
&&+(1-q\delta')\left[(1-d/q)\left(1-\log_q\frac{1}{1-d/q}\right)\right]\\
\Rightarrow \eps &\geq& q\delta'\left[(1-(d+1)/q)\log_q\frac{1}{1-(d+1)/q}\right]
+(1-q\delta')\left[(1-d/q)\log_q\frac{1}{1-d/q}\right]\\
\Rightarrow q &\geq& 2^{\frac{1}{\eps} \cdot
q\delta'\left[(1-(d+1)/q)\log\frac{1}{1-(d+1)/q}\right]+(1-q\delta')\left[(1-d/q)\log\frac{1}{1-d/q}\right]} \ge 2^{\frac{f(\delta)}{\eps}}
\end{eqnarray*}
which finishes the proof. The only step that might need some clarification is the following inequality that bounds below the right hand term with some function that only depends on $\delta$ and is non-zero in $(0, 1)$.
$$q\delta'\left[(1-(d+1)/q)\log\frac{1}{1-(d+1)/q}\right]+(1-q\delta')\left[(1-d/q)\log\frac{1}{1-d/q}\right] \ge f(\delta)$$
Note that the left hand side is the convex combination of the points that are obtained by evaluating function $g(x)=(1-x)\log\frac{1}{1-x}$ at multiples of $\frac{1}{q}$ (See Figure~\ref{fig:rounding}). We denote the left hand term by $g'(\delta, q)$.

\begin{figure}
\centering
\fullOnly{\includegraphics[scale=0.7]{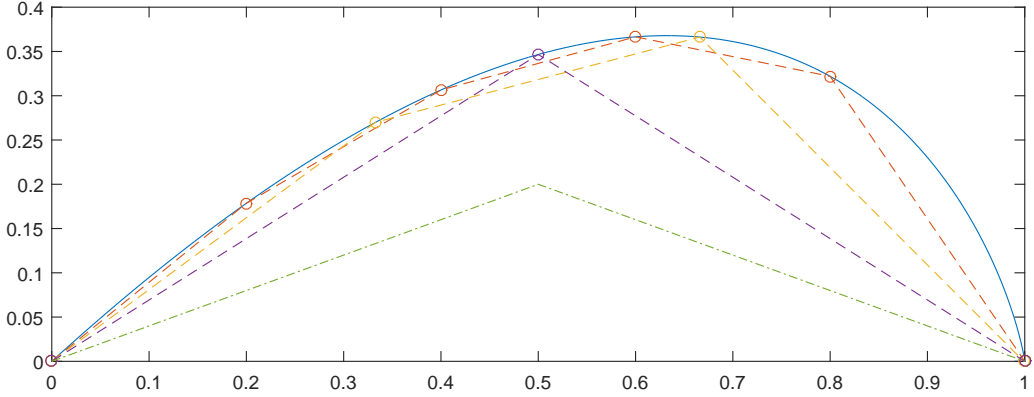}}
\shortOnly{\includegraphics[scale=0.6]{rounding.eps}}
\caption{(I) Solid line: $g(\delta)$,\quad (II) Dashed-dotted line: $f(\delta)$,\quad (III) Dashed lines: $g'(\delta, q)$ for $q=2, 3, 5$.}
\label{fig:rounding}
\end{figure}

We only need to find a function of $\delta$ that is non-zero in $(0,1)$ and is strictly smaller than such convex combinations for any $q$. One good candidate is the simple piece-wise linear function $f(\cdot)$ that consists of a segment from $(0, 0)$ to $(0.5, 0.2)$ and a segment from $(0.5, 0.2)$ to $(1, 0)$. Now, we show that for this choice of $f$, $g'(\delta, q)\geq f(\delta)$ for any $q$. 

Note that $f(\delta) \le g(\delta)$, therefore, for any $\frac{i}{q}\le\delta\le\frac{i+1}{q}$ that both $\frac{i}{q}$ and $\frac{i+1}{q}$ are smaller than $\frac{1}{2}$ or both are larger than $\frac{1}{2}$, $g'(\delta, q)\geq f(\delta)$. 
Further, for the case of $\frac{\lfloor q/2\rfloor}{q}\le \delta \le \frac{\lceil q/2\rceil}{q}$, we draw readers attention to the fact that constant $0.2$ in the definition of $f(\cdot)$ has been chosen small enough so that $g'(1/2, 3)\geq f(1/2)$ which implies that $g'(1/2, q)\geq f(1/2)$ for any $q$ and consequently $g'(\delta, q)\geq f(\delta)$ for all $\frac{\lfloor q/2\rfloor}{q}\le \delta \le \frac{\lceil q/2\rceil}{q}$ (See Figure~\ref{fig:rounding}).

\end{proof}
}
\fullOnly{\AlphabetSizeSBGapDependenceSection}

\section{Analysis of Random Codes}\label{sec:RandomCodeLowerBounds}
\global\def\AnalysisOfRandomDeletionCodes{
\subsection{Random Deletion Codes (Theorem~\ref{thm:achievabilityViaRandomCodesDeletion})}
\begin{proof}[Proof of Theorem~\ref{thm:achievabilityViaRandomCodesDeletion}]
Let $\mathcal{C}$ be a random code that maps any $x \in [1..q]^{nR}$ to some member of $[1..q]^n$ denoted by $E_{\mathcal{C}}(x)$ that has been chosen uniformly at random.
Note that in a deletion channel, for any $y=E_{\mathcal{C}}(x)$ sent by Alice, the message $z\in[1..q]^{n(1-\delta)}$ Bob receives is necessarily a subsequence of $y$. The probability of a fixed $z\in[1..q]^{n(1-\delta)}$ being a subsequence of a random $y\in[1..q]^n$ is bounded above as follows.

\begin{eqnarray}
\Pr_y\left\{z\text{ is a substring of }y\right\}
&=& \sum_{1\le a_1 < \cdots < a_{n(1-\delta)} \le n} q^{-n(1-\delta)}\left(1-\frac{1}{q}\right)^{a_{n(1-\delta)}-n(1-\delta)}\label{eqn:DeletionExpectationComputationStep1}\\
&=& \sum_{l=n(1-\delta)}^n {l \choose n(1-\delta)} q^{-n(1-\delta)}\left(1-\frac{1}{q}\right)^{l-n(1-\delta)}\nonumber\\
&\le& n\delta{n \choose n(1-\delta)} q^{-n(1-\delta)}\left(1-\frac{1}{q}\right)^{n\delta} \label{eqn:DeletionExpectationComputationStep2}\\
&=& n\delta 2^{nH(1-\delta)+o(n)} q^{-n}\left(q-1\right)^{n\delta}\nonumber\\
&=& q^{n\left(
(1-\delta)\log_q\frac{1}{1-\delta}+\delta\log_q\frac{1}{\delta}+\delta\log_q(q-1)-1+o(1)
\right)}\label{eqn:DeletionExpectationComputationStep3}
\end{eqnarray}

Step~\eqref{eqn:DeletionExpectationComputationStep1} is obtained by conditioning the probability over the leftmost appearance of $z$ in $y$ as a subsequence. $a_1, \cdots, a_{n(1-\delta)}$ denote the positions in $y$ where the first instance of $z$ is occurred. Note that in such event, $y_{a_i}$ has to be identical to $z_i$ and symbol $z_i$ cannot appear in $y[a_{i-1}+1, a_i-1]$. Hence, the probability of $z$ appearing in $a_1, \cdots, a_{n(1-\delta)}$ as the leftmost occurrence is $\left(\frac{1}{q}\right)^{n(1-\delta)}\left(1-\frac{1}{q}\right)^{a_n - n(1-\delta)}$.

To verify \eqref{eqn:DeletionExpectationComputationStep2}, we show that the largest term of the summation in the previous step is attained at $l=n$ and bound the summation above by $n$ times that term. To prove this, we simply show that the ratio of the consecutive terms of the summation is larger than one.
$$\frac{
{l \choose n(1-\delta)} q^{-n(1-\delta)}\left(1-\frac{1}{q}\right)^{l-n(1-\delta)}
}{
{l-1 \choose n(1-\delta)} q^{-n(1-\delta)}\left(1-\frac{1}{q}\right)^{l-1-n(1-\delta)}
}=
{\frac{l}{l-n(1-\delta)}\cdot\left(1-\frac{1}{q}\right)} = \frac{1-\frac{1}{q}}{1-\frac{n(1-\delta)}{l}} \ge 1$$
where the last inequality follows from the fact that
$\delta < 1-\frac{1}{q} \Rightarrow \frac{1}{q} < 1-\delta \Rightarrow \frac{1}{q} < \frac{n(1-\delta)}{l}$.


According to~\eqref{eqn:DeletionExpectationComputationStep3}, for a random code $\mathcal{C}$ of rate $R$, the probability of $l+1$ codewords $E_{\mathcal{C}}(m_1), \allowbreak E_{\mathcal{C}}(m_2), \cdots, E_{\mathcal{C}}(m_{l+1})$ for $m_1, \cdots, m_{l+1}\in [1..q]^{nR}$ containing some fixed $z\in[1..q]^{n(1-\delta)}$ as a subsequence is bounded above as follows.
$$\Pr\left\{z\text{ is a subseq. of all }E_\mathcal{C}(m_1), \cdots
\right\}\le
q^{n(l+1)\left(
(1-\delta)\log_q\frac{1}{1-\delta}+\delta\log_q\frac{1}{\delta}+\delta\log_q(q-1)-1+o(1)
\right)}
$$
Hence, by applying the union bound over all $z\in[1..q]^{n(1-\delta)}$, the probability of random code $\mathcal{C}$ not being $l$-list decodable can be bounded above as follows.
\begin{eqnarray*}
&&\Pr_{\mathcal{C}}\left\{\exists z\in[1..q]^{n(1-\delta)}, m_1, \cdots, m_{l+1}\in [1..q]^{nR}\text{ s.t. }z\text{ is a subseq. of all }E_\mathcal{C}(m_1), \cdots
\right\}\\&\le& q^{n(1-\delta)} \left(q^{Rn}\right)^{l+1} q^{n(l+1)\left(
(1-\delta)\log_q\frac{1}{1-\delta}+\delta\log_q\frac{1}{\delta}+\delta\log_q(q-1)-1+o(1)
\right)}
\end{eqnarray*}
As long as $n$'s coefficient in the exponent of $q$ is negative, this probability is less than one and drops exponentially to zero as $n$ grows. \shortOnly{This happens if and only if 
\begin{equation}
R < 1-(1-\delta)\log_q\frac{1}{1-\delta} - \delta \log_q\frac{1}{\delta}-\delta\log_q(q-1)-\frac{1-\delta}{l+1}-o(1)
\label{eqn:rateLimitDeletion}
\end{equation}
}
\fullOnly{\begin{eqnarray}
&&n(1-\delta)+Rn(l+1)+n(l+1)\left(
(1-\delta)\log_q\frac{1}{1-\delta}+\delta\log_q\frac{1}{\delta}+\delta\log_q(q-1)-1+o(1)
\right) < 0\nonumber\\
&\Leftrightarrow& R < 1-(1-\delta)\log_q\frac{1}{1-\delta} - \delta \log_q\frac{1}{\delta}-\delta\log_q(q-1)-\frac{1-\delta}{l+1}-o(1)
\label{eqn:rateLimitDeletion}
\end{eqnarray}}
Therefore, the random code $\mathcal{C}$ with any rate $R$ that satisfies \eqref{eqn:rateLimitDeletion} is list-decodable with a list of size $l$ with high probability.

We now proceed to prove the second side of Theorem~\ref{thm:achievabilityViaRandomCodesDeletion}. We will show that any family of random codes with rate $R > 1-(1-\delta)\log_q\frac{1}{1-\delta} - \delta \log_q\frac{1}{\delta}-\delta\log_q(q-1)$ is not list decodable with high probability. 

To see this, fix a received word $z\in[1..q]^{n(1-\delta)}$. Let $X_i$ be the indicator random variable that indicates whether $i$th codeword contains $z$ as a subsequence or not. Probability of $X_i=1$ was calculated in \eqref{eqn:DeletionExpectationComputationStep3} and gives that the expected number of codewords that contain $z$ is 
$$\mu=q^{nR}\cdot q^{n\left(
(1-\delta)\log_q\frac{1}{1-\delta}+\delta\log_q\frac{1}{\delta}+\delta\log_q(q-1)-1+o(1)
\right)}.$$
Note that for $R > 1-(1-\delta)\log_q\frac{1}{1-\delta} - \delta \log_q\frac{1}{\delta}-\delta\log_q(q-1)$ and large enough $n$, this number is exponentially large in terms of $n$. Further, as $X_i$s are independent, Chernoff bound gives the following.
$$\Pr\{X> \mu/2\} \ge 1-e^{-\mu/8}$$
Which implies that, with high probability, exponentially many codewords contain $z$ and, therefore, the random code is not list-decodable.
\end{proof}}
\fullOnly{\AnalysisOfRandomDeletionCodes}

\subsection{Random Insertion Codes (Theorem~\ref{thm:achievabilityViaRandomCodes})}
\begin{proof}[Proof of Theorem~\ref{thm:achievabilityViaRandomCodes}]
We prove the claim by considering a random code $\mathcal{C}$ that maps any $x\in [1..q]^{Rn}$ to some uniformly at random chosen member of $[1..q]^{n}$ denoted by $E_{\mathcal{C}}(x)$ and showing that it is possible to list-decode $\mathcal{C}$ with high probability. 

Note that in an insertion channel, the original message sent by Alice is a substring of the message received on Bob's side.
Therefore, a random code $\mathcal{C}$ is $l$-list decodable if for any $z\in[1..q]^{\left(\gamma+1\right)n}$, there are at most $l$ codewords of $\mathcal{C}$  that are subsequences of $z$.
%
For some fixed $z\in[1..q]^{\left(\gamma+1\right)n}$, the probability of some uniformly at random chosen $y \in [1..q]^n$ being a substring of $z$ can be bounded above as follows.
\begin{eqnarray*}
\Pr_y\left\{y\text{ is a subsequence of }z\right\} &\le& {(\gamma+1)n \choose n} q^{-n}\\
&=& 2^{n(\gamma+1) H(\frac{1}{\gamma+1})+o(n)} q^{-n}\\
&=& q^{n\left(\log_q(\gamma+1)+\gamma\log_q\frac{\gamma+1}{\gamma}-1+o(1)\right)}
\end{eqnarray*}
Therefore, for a random code $\mathcal{C}$ of rate $R$ and any $m_1, \cdots, m_{l+1} \in [1..q]^{nR}$ and some fixed $z \in [1..q]^{n(\gamma+1)}$:
$$\Pr\left\{E_\mathcal{C}(m_1), \cdots, E_\mathcal{C}(m_{l+1})\text{ are subsequences of } z\right\}\le q^{n(l+1)\left(\log_q(\gamma+1)+\gamma\log_q\frac{\gamma+1}{\gamma}-1+o(1)\right)}$$
Hence, using the union bound over $z \in [1..q]^{n(\gamma+1)}$, for the random code $\mathcal{C}$:
\begin{eqnarray}
&&\Pr_{\mathcal{C}}\left\{\exists z \in [1..q]^{n(\gamma+1)}, m_1, \cdots, m_{l+1}\in q^{nR}\text{ s.t. }E_\mathcal{C}(m_1), \cdots, E_\mathcal{C}(m_{l+1})\text{ are subsequences of } z\right\}\nonumber\\
&\le& q^{n(\gamma +1)} \left(q^{Rn}\right)^{l+1} q^{n(l+1)\left(\log_q(\gamma+1)+\gamma\log_q\frac{\gamma+1}{\gamma}-1+o(1)\right)}\nonumber\\
 &=& q^{n(\gamma+1)+Rn(l+1)+n(l+1)\left(\log_q(\gamma+1)+\gamma\log_q\frac{\gamma+1}{\gamma}-1+o(1)\right)}\label{eqn:InsertionFailingProb}
\end{eqnarray}
As long as $q$'s exponent in~\eqref{eqn:InsertionFailingProb} is negative, this probability is less than one and drops exponentially to zero as $n$ grows. 
\begin{eqnarray}
&&n(\gamma+1)+Rn(l+1)+n(l+1)\left(\log_q(\gamma+1)+\gamma\log_q\frac{\gamma+1}{\gamma}-1+o(1)\right) < 0\nonumber\\
&\Leftrightarrow& R < 1  - \log_q(\gamma +1) - \gamma\log_q\frac{\gamma+1}{\gamma} - \frac{\gamma +1}{l+1} +o(1)\label{eqn:rateLimit}
\end{eqnarray}
Therefore, the family of random codes with any rate $R$ that satisfies \eqref{eqn:rateLimit} is list-decodable with a list of size $l$ with high probability.	
\end{proof}

\shortOnly{\newpage}
\appendix
\begin{center}
\bfseries \huge Appendices
\end{center}

\section{Rate Needed for Unique Decoding}\label{app:UniqueDecodingRate}
In the following claim we assert that for every $\delta,\gamma > 0$ the rate of any insdel code that uniquely recovers from $\delta$-fraction deletions and $\gamma$-fraction insertions is at most $1 - (\gamma+\delta)$.

\begin{claim}
If $C$ is an insdel code of rate $R$ that can recover from $\delta$-fraction insertions and $\gamma$-fraction deletions with unique decoding, then $R \leq 1 - (\delta+\gamma)$.
\end{claim}

\begin{proof}
Let $C \subseteq \Sigma^n$ be a code with $q^k$  codewords, where $q = |\Sigma|$, for some integer $k > (1 - (\delta+\gamma))n+1$. Consider the projection of codewords to the first $k-1$ coordinates.
By the pigeonhole principle there must be two codewords $x,y \in C$ that agree on the first $k-1$ coordinates. Let $x = stu$ and $y = svw$ where $s$ is of length $k-1$, $t,v$ are of length $\gamma n$, and $u,w$ are of length $\delta n$. Now consider the string $stv$: This can be obtained from either $x$ or $y$ by first deleting the last $\delta n$ coordinates, and then inserting either $t$ (for $y$) or $v$ (for $x$). Thus no unique decoder can uniquely decode this code from $\delta$ fraction insertions and $\gamma$ fraction deletions. 
\end{proof}

The results of Haeupler and Shahrasbi~\cite{haeupler2017synchronization} in contrast show that given $\alpha$ and $\eps > 0$ there is a single code $C$ of rate $1-\alpha-\eps$ that can recover from $\delta$ fraction deletions
and $\gamma$ fraction insertions for any choice of $\delta,\gamma$ with $\gamma + \delta \leq \alpha$. 
Corollary~\ref{cor:unique-alphabetSize} shows that any such result must have exponentially large alphabet size in $\eps$ though it does not rule out the possibility that there may exist specific choices of $\gamma$ and $\delta$ for which smaller alphabets may suffice.


\shortOnly{
\AlphabetSizeSBGapDependenceSection

\section{Analysis of Random Codes (Cont.)}\label{sec:RandomCodeLowerBoundsCont}
\AnalysisOfRandomDeletionCodes
}

\newpage
\bibliographystyle{plain}
\bibliography{bibliography}

\begin{thebibliography}{10}

\bibitem{guruswami2001expander}
Venkatesan Guruswami and Piotr Indyk.
\newblock Expander-based constructions of efficiently decodable codes.
\newblock In {\em Foundations of Computer Science, 2001. Proceedings. 42nd IEEE
  Symposium on}, pages 658--667. IEEE, 2001.

\bibitem{guruswami2002near}
Venkatesan Guruswami and Piotr Indyk.
\newblock Near-optimal linear-time codes for unique decoding and new
  list-decodable codes over smaller alphabets.
\newblock In {\em Proceedings of the thiry-fourth annual ACM symposium on
  Theory of computing}, pages 812--821. ACM, 2002.

\bibitem{guruswami2003linear}
Venkatesan Guruswami and Piotr Indyk.
\newblock Linear time encodable and list decodable codes.
\newblock In {\em Proceedings of the thirty-fifth annual ACM symposium on
  Theory of computing}, pages 126--135. ACM, 2003.

\bibitem{guruswami2004linear}
Venkatesan Guruswami and Piotr Indyk.
\newblock Linear-time list decoding in error-free settings.
\newblock In {\em International Colloquium on Automata, Languages, and
  Programming}, pages 695--707. Springer, 2004.

\bibitem{guruswami2016efficiently}
Venkatesan Guruswami and Ray Li.
\newblock Efficiently decodable insertion/deletion codes for high-noise and
  high-rate regimes.
\newblock In {\em Information Theory (ISIT), 2016 IEEE International Symposium
  on}, pages 620--624. IEEE, 2016.

\bibitem{guruswami2008explicit}
Venkatesan Guruswami and Atri Rudra.
\newblock Explicit codes achieving list decoding capacity: Error-correction
  with optimal redundancy.
\newblock {\em IEEE Transactions on Information Theory}, 54(1):135--150, 2008.

\bibitem{guruswami1998improved}
Venkatesan Guruswami and Madhu Sudan.
\newblock Improved decoding of reed-solomon and algebraic-geometric codes.
\newblock In {\em Foundations of Computer Science, 1998. Proceedings. 39th
  Annual Symposium on}, pages 28--37. IEEE, 1998.

\bibitem{guruswami2011optimal}
Venkatesan Guruswami and Carol Wang.
\newblock Optimal rate list decoding via derivative codes.
\newblock In {\em Approximation, Randomization, and Combinatorial Optimization.
  Algorithms and Techniques}, pages 593--604. Springer, 2011.

\bibitem{guruswami2017deletion}
Venkatesan Guruswami and Carol Wang.
\newblock Deletion codes in the high-noise and high-rate regimes.
\newblock {\em IEEE Transactions on Information Theory}, 63(4):1961--1970,
  2017.

\bibitem{guruswami2013list}
Venkatesan Guruswami and Chaoping Xing.
\newblock List decoding reed-solomon, algebraic-geometric, and gabidulin
  subcodes up to the singleton bound.
\newblock In {\em Proceedings of the forty-fifth annual ACM symposium on Theory
  of computing}, pages 843--852. ACM, 2013.

\bibitem{GuruswamiXing2017}
Venkatesan Guruswami and Chaoping Xing.
\newblock Optimal rate list decoding over bounded alphabets using
  algebraic-geometric codes.
\newblock {\em arXiv preprint arXiv:1708.01070}, 2017.

\bibitem{haeupler2017synchronization}
Bernhard Haeupler and Amirbehshad Shahrasbi.
\newblock Synchronization strings: Codes for insertions and deletions
  approaching the singleton bound.
\newblock In {\em Proceedings of the Annual Symposium on Theory of Computing
  (STOC)}, 2017.

\bibitem{haeupler2017synchronization3}
Bernhard Haeupler and Amirbehshad Shahrasbi.
\newblock Synchronization strings: Explicit constructions, local decoding, and
  applications.
\newblock In {\em Proceedings of the Annual Symposium on Theory of Computing
  (STOC)}, 2018.

\bibitem{haeupler2017synchronization2:ARXIV}
Bernhard Haeupler, Amirbehshad Shahrasbi, and Ellen Vitercik.
\newblock Synchronization strings: Channel simulations and interactive coding
  for insertions and deletions.
\newblock {\em arXiv preprint arXiv:1707.04233}, 2017.

\bibitem{hemenway2017local}
Brett Hemenway, Noga Ron-Zewi, and Mary Wootters.
\newblock Local list recovery of high-rate tensor codes and applications.
\newblock {\em arXiv preprint arXiv:1706.03383}, 2017.

\bibitem{hemenway2015linear}
Brett Hemenway and Mary Wootters.
\newblock Linear-time list recovery of high-rate expander codes.
\newblock In {\em International Colloquium on Automata, Languages, and
  Programming}, pages 701--712. Springer, 2015.

\bibitem{kopparty2015list}
Swastik Kopparty.
\newblock List-decoding multiplicity codes.
\newblock {\em Theory of Computing}, 11(5):149--182, 2015.

\bibitem{Levenshtein65}
Vladimir Levenshtein.
\newblock Binary codes capable of correcting deletions, insertions, and
  reversals.
\newblock {\em Doklady Akademii Nauk SSSR 163}, 4:845--848, 1965.

\bibitem{mercier2010survey}
Hugues Mercier, Vijay~K Bhargava, and Vahid Tarokh.
\newblock A survey of error-correcting codes for channels with symbol
  synchronization errors.
\newblock {\em IEEE Communications Surveys \& Tutorials}, 12(1), 2010.

\bibitem{mitzenmacher2009survey}
Michael Mitzenmacher.
\newblock A survey of results for deletion channels and related synchronization
  channels.
\newblock {\em Probability Surveys}, 6:1--33, 2009.

\bibitem{schulman1999asymptotically}
Leonard~J. Schulman and David Zuckerman.
\newblock Asymptotically good codes correcting insertions, deletions, and
  transpositions.
\newblock {\em IEEE transactions on information theory}, 45(7):2552--2557,
  1999.

\bibitem{sloane2002single}
Neil~JA Sloane.
\newblock On single-deletion-correcting codes.
\newblock {\em Codes and designs}, 10:273--291, 2002.

\bibitem{wachter2017list}
Antonia Wachter-Zeh.
\newblock List decoding of insertions and deletions.
\newblock {\em IEEE Transactions on Information Theory}, 2017.

\end{thebibliography}

\end{document}